\documentclass[12pt,draftclsnofoot,journal,onecolumn]{IEEEtran}

\ifCLASSINFOpdf
\else
\fi
\usepackage{graphicx}
\usepackage{epstopdf}
\usepackage[latin1]{inputenc}
\usepackage{amsmath,amssymb,amscd,latexsym,dsfont}
\usepackage[english]{babel}
\usepackage{tabularx,cite}
\usepackage{graphicx}
\usepackage{amsmath,bm}
\usepackage{lipsum}
\usepackage{algorithm,algpseudocode}
\usepackage{float}
\usepackage{multicol}
\usepackage{graphicx}
\usepackage{caption}
\usepackage{psfrag}
\usepackage{epstopdf}
\usepackage{comment,psfrag,subfigure,enumerate}
\usepackage{algorithm}
\usepackage{pifont}

\usepackage{color}
\usepackage{amsthm}
\usepackage{bigints}
\usepackage{lipsum}
\usepackage{relsize}
\usepackage{graphicx}
\usepackage{epstopdf}
\usepackage{calligra}
\theoremstyle{definition}
\usepackage{amsmath,bm}
\usepackage{mathtools}
\newtheorem{theorem}{Theorem}

\usepackage{algorithm}

\usepackage{amsmath}
\usepackage[cmintegrals]{newtxmath}
\begin{document}

\title{Joint Sum Rate And Error Probability Optimization: Finite Blocklength Analysis}
%
%
%

\author{Mahdi Haghifam, Mohammad Robat Mili, Behrooz Makki, Masoumeh Nasiri-Kenari, Tommy Svensson.  \thanks{Mahdi Haghifam, Mohammad Robat Mili and Masoumeh Nasiri-Kenari are with Electrical Engineering Department, Sharif University of Technology,
Tehran, Iran. Emails: $\text{haghifam\_mahdi}@\text{ee.sharif.edu}$, $\text{mohammad.robatmili}@\text{ieee.org}$, $\text{mnasiri}@\text{sharif.edu}$. Behrooz Makki and Tommy Svensson are with Chalmers University of Technology, Gothenburg, Sweden. Emails:$\{\text{behrooz.makki,
tommy.svensson}
\}@\text{chalmers.se}$.
} 
}

\date{\vspace{-5ex}}
\maketitle

\begin{abstract}
We study the tradeoff between the
sum rate and the error probability in downlink of wireless networks. Using the recent results on the achievable rates of finite-length codewords, the problem is cast as a joint optimization of the network sum rate and the per-user error probability. Moreover, we develop an efficient algorithm based on the divide-and-conquer technique to simultaneously maximize the network sum rate and minimize the maximum users' error probability and to evaluate the effect of the codewords length on the system performance. The results show
that, in delay-constrained scenarios, optimizing the per-user error probability plays a key role in achieving high throughput.
\end{abstract}


%
\IEEEpeerreviewmaketitle

\section{Introduction}
The fifth generation of wireless communication (5G) must support novel traffic types for which low latency, high data rate, and ultra reliability are of interest. Particularly, in many applications
such as vehicle-to-vehicle and vehicle-to-infrastructure
communications for traffic efficiency/safety or real-time video
processing for augmented reality, the codewords are required to
be short (in the order of $\sim$ 100 channel uses) with stringent requirements on the latency and reliability \cite{tutor2}. Therefore, it is interesting to optimize the performance of wireless networks
in the presence of finite-length codewords.\par
 In 2010, \cite{poly_main} presented accurate information-theoretic approximations for the achievable rates of finite blocklength codes. Using \cite{poly_main}, the performance of wireless networks with short packets has been studied in various papers, for the cases with cognitive radio \cite{cog_mak}, relay networks \cite{li2016throughput}, hybrid automatic repeat request technique \cite{harq}.\par
In this letter, we consider a wireless network with an access point (AP) serving multiple users. Using short packets, the AP transmits packets in the downlink to the users, which have different target error probability requirements. Particularly, using the recent results of \cite{poly_main}, we propose a joint sum rate and per-user error probability optimization problem and investigate the effect of the codeword length on the system performance. To solve the joint sum rate and per-user error probability optimization problem, we develop a low-complexity two-level algorithm based on the divide-and-conquer approach. Also, we derive a closed-form expression for the optimal per-user error probability (Theorem \ref{porp:er}). Finally, we find an efficient and close-to-optimal power allocation algorithm, in terms of sum rate and error probability, based on the augmented Lagrange method \cite{bert} (Algorithm 1).    
\par
The simulation and analytical results show that 1) our proposed algorithm can reach (almost) the same performance as in the exhaustive search-based approach with considerably less implementation complexity (Figs. 1 and 2). Also, 2) the throughput is sensitive to the length of short packets while its sensitivity to the packet length decreases for long packets (Fig. 2b). Finally, 3) optimal error probability assignment with water-filling (WF) power allocation achieves higher throughput, compared to using the optimal power allocation with equal error probability assignment (Figs. 1, 2a).
\section{System Model}
We consider a downlink communication model with $N$ single-antenna users which are served by a single-antenna AP. It is assumed that each user is allocated an orthogonal channel to the AP, for instance they could be separated in the frequency or time domain. Let us denote the instantaneous channel gain between the AP and the $i$-th user, $i=1, \ldots, N,$ by $g_i$. The channel gain $g_i$ can be expressed as $g_i=\bar{g_i} \theta_i $
, where $\theta_i$ represents the small scale fading and $\bar{g_i}$
is the average channel gain, obtained by considering the path loss
effects and shadowing. Thus, with $i$-th user located at distance $d_i$ from the AP, we have $\bar{g_i}=\kappa_i d_i^{-\delta_i}$ where $\kappa_i$ is the signal power gain at distance 1 meter from the AP and $\delta_i$  is the path loss exponent. Moreover, the power allocated by the AP to the signal of user $i$ is denoted by $p_{i}$. Thus, the instantaneous signal-to-noise ratio (SNR) received by user $i$ is $\gamma_{i}=\frac{p_{i}g_{i}}{\sigma^2}$, where $\sigma^2$ is the noise power density. 
We characterize the network performance when the AP employs packets of short length. Specifically, the user $i$'s message is encoded into the packets of length $L$ and transmitted with power $p_i$. In this way, the maximum achievable information rate in nats per channel use (npcu) for user $i$ which can be decoded with  
block error probability no greater than $\epsilon_i$ is given by\cite[Thm. 1]{tan2015third}
\begin{equation}\label{eq:rate_fbl}
r_{i}=\log\left(1+\gamma_i p_i\right)-\sqrt{\frac{1}{L}\left(1-\frac{1}{\left(1+\gamma_i p_i\right)^2}\right)}Q^{-1}\left(\epsilon_i\right) +\frac{\log\left(L\right)}{L} .
\end{equation}
In (\ref{eq:rate_fbl}), the achievable rate increases unboundedly as the error probability tends towards one. On the other hand, the rate decreases significantly in the cases with strict error probability requirements, i.e., small $\epsilon_i$'s. Also, the achievable rate increases with the signal length $L$ monotonically and letting $L\to\infty$, the achievable rate (\ref{eq:rate_fbl}) converges to Shannon's capacity formula in the cases with asymptotically long codewords.\par  
Motivated by the tradeoff between the achievable rates and the error probability in (\ref{eq:rate_fbl}), we consider a joint sum rate maximization and per-user error probability minimization problem. Assuming perfect channel state information (CSI) at the AP, we study a multi-objective optimization problem
\begin{subequations}
\label{opt:main0}
\begin{align}
    \underset{\bm{\epsilon},\bm{p}}{\text{maximize}}
        & \quad \sum_{i=1}^{N}\log\left(1+\gamma_ip_i\right)- \sqrt{\frac{1}{L}\left(1-\frac{1}{\left(1+\gamma_ip_i\right)^2}\right)}  Q^{-1}\left(\epsilon_i\right) \label{cost_sr} \\
      \underset{\bm{\epsilon},\bm{p}}{\text{minimize}}
        & \quad   \max\{\epsilon_1,\hdots,\epsilon_N\} \label{cost_er}  \\
    \text{subject to} 
        & \quad 0\leq\epsilon_i \leq \varepsilon_{\text{max},i} \quad \forall i\in \{1,\hdots,N\} ,\label{const_er}\\
        & \quad 0\leq p_i \quad \forall i\in \{1,\hdots,N\}, \label{const_pow1}\\
         & \quad \sum_{i=1}^{N} p_i \leq P_{\text{max}}, \label{const_pow2}
\end{align}
\end{subequations}
where $\bm{\epsilon}=[\epsilon_1,\hdots,\epsilon_N]$ and $\bm{p}=[p_1,\hdots,p_N]$. Also, in (\ref{cost_er}), the goal is to minimize the maximum error probability of the users. Then, in (\ref{const_er}), $\varepsilon_{\text{max},i}$ is the maximum error probability constraint of user $i$ which indicates the supporting quality of service (QoS) requirement of user $i$. Also, the total power constraint of the AP is denoted by $P_{\text{max}}$. In this way, (\ref{opt:main0}) is of interest in emerging applications of 5G calling for heterogeneous QoS requirements on data rate and reliability. For instance, massive machine-type communication and ultra-reliable and low latency communication scenarios demand short-length packet exchange with stringent requirement on reliability at moderately low rate\cite{tutor2}. The requirements of the aforementioned services illustrate how the
framework of (\ref{opt:main0}) can be utilized to balance conflicting performance objectives, namely, sum rate and maximum error probability. Moreover, as seen in the following, our discussions are well applicable to the cases when optimizing the network sum throughput, which is defined as the product of the rate and the successful message decoding probability. However, as opposed to throughput optimization, (\ref{opt:main0}) is flexible in optimizing the rates and error probabilities individually based on the QoS requirements. \par
Depending on the number of users, there may be no closed-form solution for (\ref{opt:main0}). Thus, we follow the same method as in \cite{yu2016tradeoff} to convert the problem into a single objective optimization using the weighted sum method while normalizing the objectives. Also, as seen in Section \ref{sec:sim}, our proposed sub-optimal approach can reach (almost) the same performance as in the optimal exhaustive search-based scheme. With no loss of generality, we assume $\varepsilon_{\text{max},1}\leq \varepsilon_{\text{max},2} \leq \hdots \leq \varepsilon_{\text{max},N} < \frac{1}{2}$. Also, to guarantee a consistent comparison between the objectives in (\ref{cost_sr}) and (\ref{cost_er}), we normalize them as
\begin{align}
&U_1\left(\bm{p},\bm{\epsilon}\right)=\frac{\left[\sum_{i=1}^{N}\log\left(1+\gamma_ip_i\right)- \sqrt{\frac{1}{L}\left(1-\frac{1}{\left(1+\gamma_ip_i\right)^2}\right)}  Q^{-1}\left(\epsilon_i\right)\right]}{\text{SR}_{\infty}}\\
 &U_2\left(\bm{\epsilon}\right)=\frac{\varepsilon_{\text{max},N}-\max\{\epsilon_1,\hdots,\epsilon_N\}}{\varepsilon_{\text{max},N}},
\end{align} 
 where $\text{SR}_{\infty}$ is a normalization factor that can be found by plugging the water-filling power allocation expression into the Shannon's capacity formula which provides an upper bound for (\ref{cost_sr}). Then, we use the weighted sum method to rewrite (\ref{opt:main0}) as the single-objective optimization problem
\begin{subequations}
\label{opt:main_after}
\begin{align}
    \underset{\bm{\epsilon},\bm{p}}{\text{maximize}}
        & \quad \omega U_1\left(\bm{p},\bm{\epsilon}\right)+\left(1-\omega\right) U_2\left(\bm{\epsilon}\right) \\
    \text{subject to} 
        &\quad  (\text{\ref{const_er}})-(\text{\ref{const_pow2}}).
\end{align}
\end{subequations}
Here, $0 \leq \omega \leq 1$ is the weighting parameter. Note that, with $\omega$ ranging from $0$ to $1$, scenarios with strict rate requirements and relaxed error probability requirements to scenarios with low rate requirements and ultra-reliability are addressed.
\section{Proposed Algorithm}\label{sec:alg}
\vspace{-1mm}
The optimization problem (\ref{opt:main_after}) belongs to the class of non-convex problems which has a multi-modal objective function, so finding its global optimal solution is computationally infeasible. For this reason, we apply the primal decomposition approach \cite{palomar2006tutorial} to optimize $\bm{\epsilon}$ and $\bm{p}$ separately. In this way, to solve (\ref{opt:main_after}), we use the following iterative approach
\begin{equation*}
\begin{aligned}
\underbrace{\bm{p}[0]}_{\text{initialization}}&\rightarrow \bm{\epsilon}[1]\rightarrow\bm{p}[1]\rightarrow \hdots \underbrace{\bm{\epsilon}[T]\rightarrow\bm{p}[T]}_{\text{optimal solution}},
\end{aligned}
\end{equation*}
where $\bm{\epsilon}[t]$ and $\bm{p}[t]$ are the optimal error probability and power allocation vectors at iteration $t$, and $T$ is the maximum number of iterations considered by the network designer. The details of our proposed optimization approach are as follows.

\subsection{Error Probability Optimization For A Given Power Allocation}
Here, for a given power allocation $\bm{p}^{\star}\left[t-1\right]:= \bm{p}$, we find the optimal error probabilities of each user at iteration $t$ denoted by $\bm{\epsilon}^{\star}[t]$. Setting $z=\max\{\epsilon_1,\hdots,\epsilon_{N}\}$ and assuming a given power allocation, (\ref{opt:main_after}) is rephrased as
\begin{subequations}
\label{opt:sub_er}
\begin{align}
    \underset{\bm{\epsilon},z}{\text{minimize}}
        & \quad \frac{\omega}{\text{SR}_{\infty}} \sum_{i=1}^{N} \sqrt{\frac{1}{L}\left(1-\frac{1}{\left(1+\gamma_ip_i\right)^2}\right)}  Q^{-1}\left(\epsilon_i\right)+\frac{1-\omega}{\varepsilon_{\text{max},N}}z \label{subprob1:obj}  \\
    \text{subject to} 
        & \quad  \epsilon_i\leq z \quad \quad \forall i\in \{1,\hdots,N\} \label{const:sub11} \\
        & \quad 0\leq z \leq \epsilon_{\text{max},N} \label{const:sub12}\\
        & \quad 0\leq\epsilon_i \leq \varepsilon_{\text{max},i} \quad \forall i\in \{1,\hdots,N\} \label{const:sub13}
\end{align}
\end{subequations}
Theorem \ref{porp:er} gives a closed-form expression for the optimal error probability assignment of each user in terms of (\ref{opt:sub_er}).
\begin{theorem}\label{porp:er}
The optimal error probabilities of the users are given by
\begin{equation} \label{eq:opt_err_nahaie}
\bm{\epsilon}^{\star}=\begin{cases}\vspace{1mm}
\Big[\varepsilon_{\text{max},1},\hdots,\varepsilon_{\text{max},k-1},\underbrace{\beta_{k},\hdots,\beta_{k}}_{N-k+1 \ \text{times}} \Big] & \beta_{k} \in \mathcal{I}_{k}\\ 
\left[\varepsilon_{\text{max},1},\hdots,\varepsilon_{\text{max},N}\right] & \text{otherwise}
\end{cases},
\end{equation}
for $k=1,\hdots,N$. Here, we define the intervals $\mathcal{I}_{k}=(\varepsilon_{\text{max},k-1},\varepsilon_{\text{max},k}]$ with $\mathcal{I}_{1}=(0,\varepsilon_{\text{max},1}]$, and \\ $\beta_{k}= Q\Bigg(\sqrt{\smash[b]{2\log\Big(\frac{\sqrt{L}\left(1-\omega\right) \text{SR}_{\infty}}{ \varepsilon_{\text{max},N}  \omega\sqrt{2\pi}\sum_{i=k}^{N}\frac{\sqrt{\gamma_i^2 p_i^2+2\gamma_i p_i}}{1+\gamma_i p_i} }\Big)}}\Bigg)$. 
\end{theorem}
\begin{proof}
Since the constraints in (\ref{const:sub11})-(\ref{const:sub13}) are affine functions in $\bm{\epsilon}$ and $z$, it is sufficient to prove that the objective function in (\ref{subprob1:obj}) is convex. The second derivative of $Q^{-1}\left(x\right)$ is given by $\frac{\text{d}^2 Q^{-1}\left(x\right)}{\text{d}x^2}=2\pi Q^{-1}\left(x\right)\exp\left(\left( Q^{-1}\left(x\right)\right)^2\right) > 0$ if $x < \frac{1}{2}$. Therefore, considering the fact that $\epsilon_i \leq \epsilon_{\text{max},i} \leq \frac{1}{2}$, the objective function in (\ref{subprob1:obj}) is a sum of convex functions and an affine function, i.e., $z$. Hence, (\ref{opt:sub_er}) is a convex optimization problem, and the optimal solution can be found by considering Karush-Kuhn-Tucker (KKT) conditions. Thus, we write the Lagrangian function of (\ref{opt:sub_er}) as
\begin{equation}\nonumber
\mathcal{L}\left(\bm{\epsilon},z,\bm{\lambda},\bm{\nu},\eta\right)= \frac{\omega}{\text{SR}_{\infty}} \sum_{i=1}^{N} \sqrt{\frac{1}{L}\left(1-\frac{1}{\left(1+\gamma_ip_i\right)^2}\right)}  Q^{-1}\left(\epsilon_i\right)+
\end{equation}
\begin{equation}
\frac{1-\omega}{\varepsilon_{\text{max},N}}z-\sum_{i=1}^{N}\lambda_i\left(z-\epsilon_i\right)-\sum_{i=1}^{N}\nu_i\left(\epsilon_{\text{max},i}-\epsilon_i\right)-\eta\left(\varepsilon_{\text{max},N}-z\right),
\end{equation}
where $\bm{\lambda}=\left[\lambda_1,\hdots,\lambda_N\right] \succeq 0$, $\eta\geq 0$, and $\bm{\nu}=\left[\nu_1,\hdots,\nu_N\right] \succeq 0$ are dual variables associated with constraints in (\ref{const:sub11}), (\ref{const:sub12}), and (\ref{const:sub13}), respectively. According to the KKT conditions, the optimal solution, which is denoted by $\bm{\epsilon}^{\star}$ and $z^{\star}$, should satisfy
\begin{subequations}
\label{KKT:COND}
\begin{align}
&\frac{\partial \mathcal{L}}{\partial \epsilon^{\star}_i}= -\frac{\omega}{\text{SR}_{\infty}}\sqrt{\frac{1}{L}\left(1-\frac{1}{\left(1+\gamma_ip_i\right)^2}\right)} \sqrt{2\pi} \exp\left(\frac{\left(Q^{-1}\left(\epsilon^{\star}_i\right)\right)^2}{2}\right) \nonumber \\
&+\lambda^{\star}_i+\nu^{\star}_i=0, \label{KKT:derv1}\\
 &\frac{\partial \mathcal{L}}{\partial z^{\star}} = \frac{1-\omega}{\varepsilon_{\text{max},N}}-\sum_{i=1}^{N}\lambda^{\star}_i+\eta^{\star}=0, \label{KKT:derv2} \\
&\lambda^{\star}_i\left(z^{\star}-\epsilon^{\star}_i\right)=0,\label{KKT:eps1}\\
&\nu^{\star}_i\left(\varepsilon_{\text{max},i}-\epsilon^{\star}_i\right)=0,\label{KKT:eps2}\\
&\eta^{\star}\left(\varepsilon_{\text{max},N}-z^{\star}\right)=0\label{KKT:eps5},\\
&\text{(\ref{const:sub11})-(\ref{const:sub13})}\label{KKT:eps4}.
\end{align}
\end{subequations}
In (\ref{KKT:derv1}), we have used $\frac{\text{d} Q^{-1}\left(x\right)  }{\text{d} x}=-\sqrt{2\pi}\exp\left(\frac{\left(Q^{-1}\left(x\right)\right)^2}{2}\right)$. From (\ref{KKT:eps1}) and (\ref{KKT:eps2}), it can be verified that $\epsilon_i^{\star}$ is equal to either $z^{\star}$ or $\varepsilon_{\text{max},i}$; otherwise, $\lambda^{\star}_i$ and $\nu^{\star}_i$ must be equal to zero which contradict with (\ref{KKT:derv1}). Assume  $z^{\star} < \varepsilon_{\text{max},N}$ and $z^{\star}\in \mathcal{I}_k$, so according to (\ref{KKT:eps5}), we have $\eta^{\star}=0$. Note that, for $ 1 \leq i \leq k-1$, $\epsilon_{i}^{\star}$ must be equal to $\varepsilon_{\text{max},i}$ since we have $z^{\star} > \varepsilon_{\text{max},i}$. Also, due to the fact that $z^{\star} < \varepsilon_{\text{max},i}$, it can be inferred that $\epsilon_{i}^{\star}=z^{\star}$ for $ k \leq i \leq N$. Thus, in the cases with $z^{\star}\in \mathcal{I}_k$, we have  $\bm{\epsilon}^{\star}=\left[\varepsilon_{\text{max},1},\hdots,\varepsilon_{\text{max},k-1},\underbrace{z^{\star},\hdots,z^{\star}}_{\text{$N-k+1$ times}}\right]$. Then, from (\ref{KKT:eps1}) and (\ref{KKT:eps2}), it can be concluded that $\lambda^{\star}_{i}=0$ for $1 \leq i \leq k-1$ and $\nu^{\star}_i=0$ for $k \leq i \leq N$. In this way, (\ref{KKT:derv1}) can be expressed as
\begin{equation}\label{eq:lambda_megh}
\frac{\omega}{\text{SR}_{\infty}}\sqrt{\frac{1}{L}\left(1-\frac{1}{\left(1+\gamma_ip_i\right)^2}\right)} \sqrt{2\pi} \exp\left(\frac{\left(Q^{-1}\left(z^{\star}\right)\right)^2}{2}\right)=\lambda^{\star}_i,
\end{equation}
for $k \leq i \leq N$. Also, from (\ref{KKT:derv2}) and (\ref{eq:lambda_megh}), we have
\begin{equation}\label{eq:find_b}
\begin{aligned}
&\frac{1-\omega}{\varepsilon_{\text{max},N}}=\sum_{i=k}^{N}\lambda^{\star}_i.
\end{aligned}
\end{equation}
Plugging (\ref{eq:lambda_megh}) into (\ref{eq:find_b}), the upper branch of (\ref{eq:opt_err_nahaie}) is found. In this way, depending on $z^{\star}$ being in each region $\mathcal{I}_{k}$, the closed-form solution for $\bm{\epsilon}^{\star}$ is provided. Then, given $z^{\star}=\varepsilon_{\text{max},N}$, it is straightforward
to show that the objective function in (\ref{subprob1:obj}) is a decreasing function in each $\epsilon_i$, so the lower branch of (\ref{eq:opt_err_nahaie}) provides the optimal solution. Note that because of the strict convexity of (\ref{subprob1:obj}), there is an optimal solution for $\bm{\epsilon}$ which is found by searching in $N+1$ branches of (\ref{eq:opt_err_nahaie}).
\end{proof}
\vspace{-4mm}
\subsection{Optimal Power Allocation For A Given Error Probability} 
Consider a given $\bm{\epsilon}[t] := \bm{\epsilon}$. Then, (\ref{opt:main_after}) is relaxed to 
\begin{subequations}
\label{opt:sub_pow}
\label{opt:subprolme2}
\begin{align}
    \underset{\bm{p}}{\text{maximize}}
        & \quad \omega U_1\left(\bm{p},\bm{\epsilon}\right) \label{subprob2:obj}  \\
    \text{subject to} 
        & \quad \text{(\ref{const_pow1})-(\ref{const_pow2})}.
\end{align}
\end{subequations}
Since the function in (\ref{subprob2:obj}) is non-concave in $\bm{p}$, problem (\ref{opt:sub_pow}) belongs to the class of non-convex optimization problems. In a non-convex problem, there is a nonzero duality gap between primal and dual problems. Here, we use the \textit{augmented Lagrange approach} \cite[sec 4]{bert} to deal with this non-convex optimization which reduces the duality gap by augmenting a penalty-like quadratic term added to the Lagrangian function. In
\cite[Sec 4.2]{bert}, it has been proved that the augmented Lagrangian is locally convex when the penalty parameter is sufficiently large.  In contrast to the penalty functions approach, the
augmented Lagrangian function largely preserves smoothness
and does not require an asymptotically large penalty parameter for the method to converge, meaning that the penalization is exact. Augmented Lagrangian algorithms are based on successive maximization of the augmented Lagrangian function in which the multiplier estimates and penalty parameter are fixed in each iteration and then updated
between iterations. Applying the augmented Lagrangian method on (\ref{opt:sub_pow}), which eliminates the constraints and adds them to the objective function, gives the augmented Lagrangian function
\begin{equation}
\begin{aligned}
&\mathcal{L}_{\mu,\zeta}(\bm{p})=
\frac{\omega}{\text{SR}_{\infty}}\left[\sum_{i=1}^{N}\log\left(1+\gamma_ip_i\right)- \sqrt{\frac{1}{L}\left(1-\frac{1}{\left(1+\gamma_ip_i\right)^2}\right)}  Q^{-1}\left(\epsilon_i\right)\right] \\
&-\frac{1}{2\mu} \left[\left(\max\left\lbrace 0,\zeta-\mu\left(P_{\text{max}}-\sum_{i=1}^{N}p_i\right)\right\rbrace \right)^2-\zeta^2\right],
\end{aligned}
\end{equation}
where $\mu$ is  a positive coefficient denoting the penalty parameter and
$\zeta$ is the Lagrangian dual variable associated with (\ref{const_pow2}). Then, at stage $l$ of the power allocation problem, we solve
\begin{equation}
    \underset{\bm{p}}{\text{maximize}}
       \quad \mathcal{L}_{\mu^{(l)},\zeta^{(l)}}\left(\bm{p}\right), \label{ag:}  
\end{equation}
 which approximates (\ref{opt:sub_pow}) to find the power allocation at iteration $ l+1 $ denoted by $\bm{p}^{(l+1)}$.  Moreover, variables $\zeta^{(l)}$ and $\mu^{(l)}$ are updated according to
\begin{equation}
\begin{aligned}
\zeta^{(l+1)}&=\max\left\{  0,\zeta^{(l)}-\mu^{(l)}\left(  P_{\max}-\sum_{i=1}^{N} p_{i}^{(l)}
\right)  \right\}, \\ 
\mu^{(l+1)}&=2\mu^{(l)}, \label{up:zeta_eta}%
\end{aligned}
\end{equation}
respectively. In this way, as $\mu^{(l)}$ increases, the
violations introduced by constraints are penalized more severely so that the maximizer of the penalty function in (\ref{ag:}) gives the results closer to the feasible region. In \cite[Sec 4.2]{bert}, it has been shown that while the constraints are nonlinear, the convergence rate of the augmented Lagrangian method is linear. 
\par
The iterative joint error probabilities and power allocation
algorithm is summarized in \textbf{Algorithm \ref{CHalgorithm}}. In order to analyze the complexity order of the proposed algorithm, we note that the optimal error probabilities can be found by (\ref{eq:opt_err_nahaie}) with the complexity $\mathcal{O}\left(N\right)$. Also, the complexity of the power allocation at each  iteration is $\mathcal{O}\left(N^2\right)$. Thus, the complexity of Algorithm \ref{CHalgorithm} is $\mathcal{O}\left(N^2\right)$ + $\mathcal{O}\left(N\right)= \mathcal{O}\left(N^2\right).$
\begin{algorithm}
\caption{Error Probabilities Assignment and Power Allocation}
\label{CHalgorithm}
\begin{algorithmic}[1]
\State For every given $\omega$, $P_{\text{max}}$, $\{\varepsilon_{\text{max},1},\hdots,\varepsilon_{\text{max},N}\}$, $\mu^{(0)}$, and $\zeta^{(0)}$.
\State Initialize: $\bm{p}[0]$ and $t=0$.
\While{$\{\epsilon^{\star}[t] \}$ converges} 
\State Calculate $\bm{\epsilon}^{\star}[t+1]$ via (\ref{eq:opt_err_nahaie}) with given $\bm{p}^{\star}[t]$.
\State Initialize: $l=0$.
\While{$\{p^{(l)} \}$ converges} 
\State Calculate $\bm{p}^{(l+1)}$ via (\ref{ag:}) with given $\bm{\epsilon}^{\star}[t+1]$, $\zeta^{(l)}$, and $\mu^{(l)}$.
\State Update $\zeta^{(l+1)}$ and $\mu^{(l+1)}$ via (\ref{up:zeta_eta}).
\State $l=l+1$.
\EndWhile
\State $\bm{p}^{\star}[t+1]=\bm{p}^{(l)}$.
\State $t=t+1$
\EndWhile
\end{algorithmic}
\end{algorithm}
\section{Numerical Results and Conclusion}\label{sec:sim}
Here, we study the trade-off between the sum rate and maximum error probability. We set the noise power $\sigma^2 = 1$, the number of users $N=4$, and the error probability constraints 
$\bm{\varepsilon_{\text{max}}}=[10^{-5},5\times 10^{-5},10^{-4},5\times 10^{-4}]$ which are often assumed for vehicular-to-vehicular communications\cite{tutor2}. Also, it is assumed that the users are equidistant from the AP. Also, we consider
Rayleigh fading with mean 1. Finally, we set $\mu^{(0)}=1$ and $\zeta^{(0)}=0.15$ in Algorithm 1. 
For the numerical results, we consider the cases with $ L \geq 100$ channel uses, for which the
approximation (\ref{eq:rate_fbl}) is tight enough \cite{poly_main}. Also, we compare our method with three baseline algorithms: 1) WF-based power allocation with the
error probabilities of all users set to the minimum of the required error probabilities, called minmax error
probability assignment, i.e., $\bm{\epsilon}=[\varepsilon_{\text{max},1},\hdots,\varepsilon_{\text{max},1}]$, 2) the proposed method for power allocation with the minmax error probability assignment, and 3) equal power allocation with the proposed method for the error probabilities assignment. Finally, the results are obtained by averaging over $10^4$ different channel realizations. 
 \par  
Figure 1 shows the the tradeoff between the sum rate and the error probability for different algorithms with $P_{\text{max}}=6$ dB and $L=200$ channel uses. As a performance metric, we define the sum throughput as
\begin{equation}
\mathcal{T}=\sum_{i=1}^{N}r_i\left(1-\epsilon_i\right),
\end{equation}
where the user $i$'s codeword rate and error probability are given by $r_i$ and $\epsilon_i$, respectively.
Then, Fig. 2a demonstrates the sum throughput versus the AP's
total power constraint $P_{\text{max}}$ by setting $L=100$ channel uses and $\omega=0.9$. Finally, Fig. 2b evaluates the effect of the codeword length on the sum throughput. The results lead to the following conclusions:\\
\begin{figure}
\centering
\includegraphics[width=5in, height=3in]{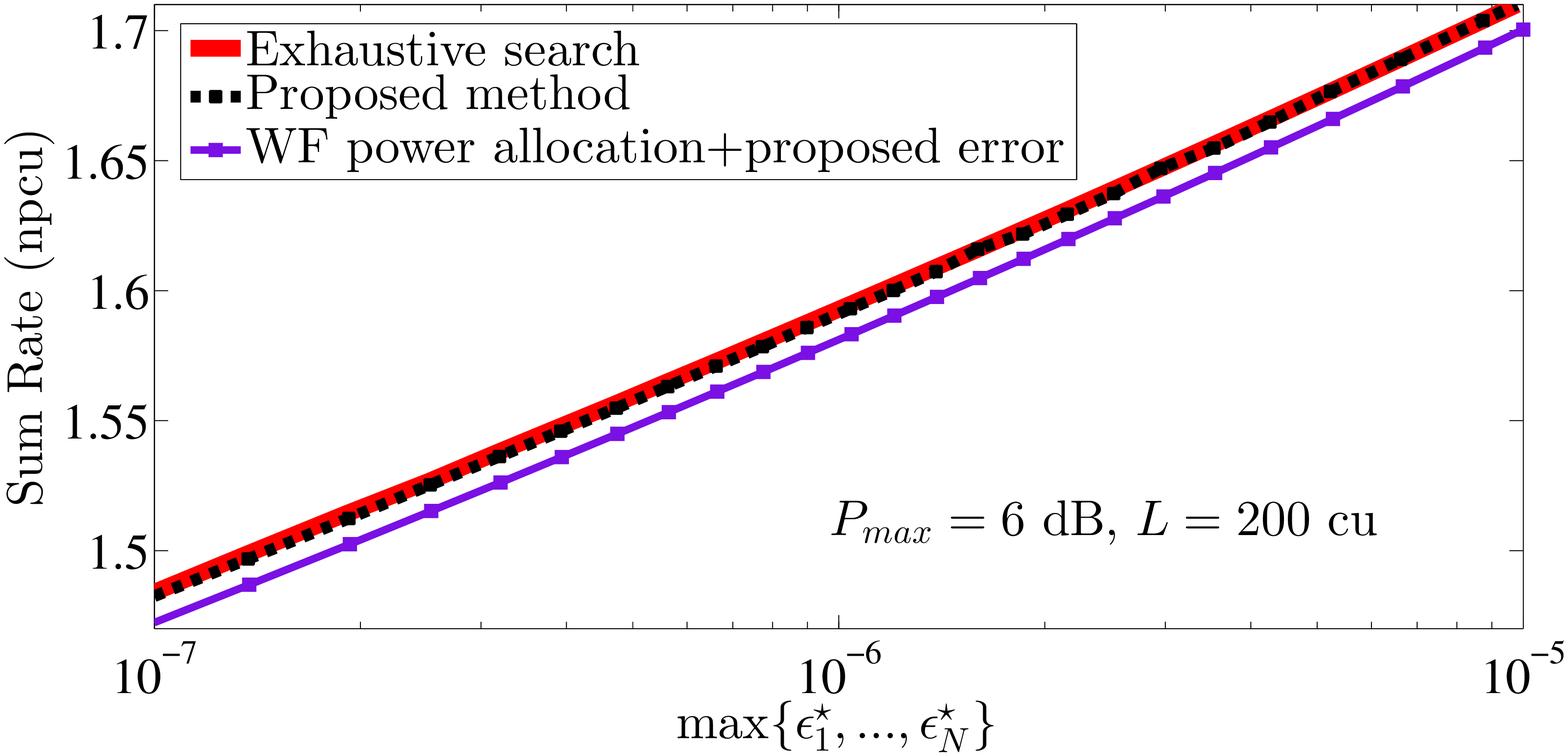}
\caption{Sum rate vs maximum error probability. }
\vspace{-5mm}
\end{figure}
$\bullet$  The scheme with the WF power allocation and the error probability assignment based on Theorem 1 achieves the tradeoff region close to the proposed method optimizing both the error probability and the power allocation (Fig. 1).\\
$\bullet$ For short codewords, the throughput is remarkably affected by the length of the codeword. However, the effect of increasing the codeword length  decreases for long codewords (Fig. 2b). Also, optimal error probability assignment with WF power allocation achieves higher throughput compared to optimal power allocation with equal error probability assignment. Moreover, the performance of WF with minmax error probability assignment is close to that of the scheme with the proposed power allocation with minmax error probability assignment (Figs. 2a and 2b). \\
$\bullet$   For short codeword (say,
$L\leq 1000$ channel uses), 
the proposed algorithm leads to considerable throughput improvement in comparison with other schemes. For instance, when $L=100$ cu, the performance of the proposed method has $66 \%$ of improvement. However, the performance difference of the schemes decreases in the cases with long codewords.\\
$\bullet$ Finally, as observed in Figs. 1, 2a, and 2b, the gap between the developed algorithm and the exhaustive search-based algorithm diminishes by increasing $P_{\text{max}}$ or $L$.  Thus, our proposed algorithm can be effectively applied to jointly optimize the sum rate and the error probability of multi-user networks in delay-constrained applications.
\begin{figure}
\centering
\includegraphics[width=5in, height=3in]{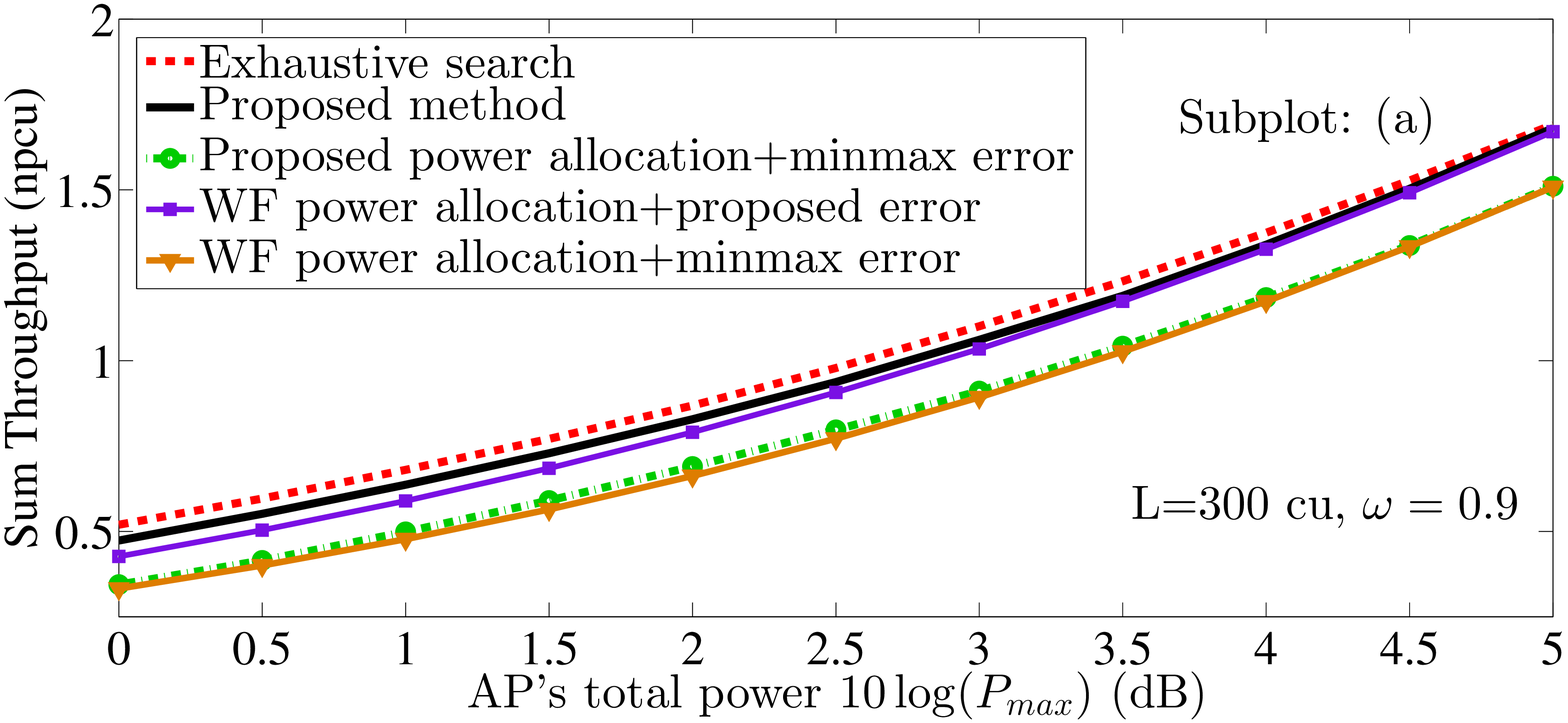}
\includegraphics[width=5.18in, height=3in]{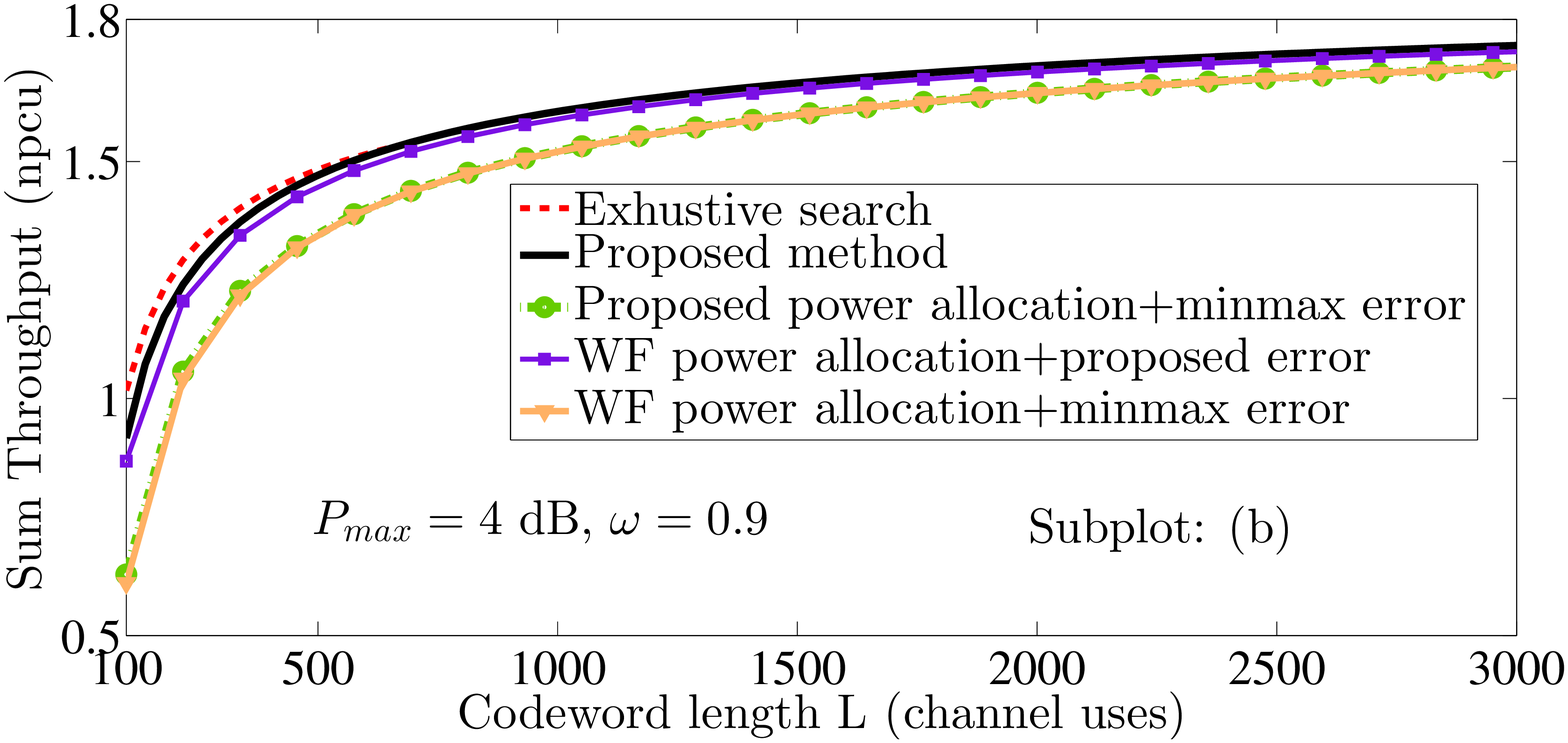}
\caption{Sum throughput of the considered algorithms. Subplot (a): Sum throughput vs the total power constraint of the AP. Subplot (b): Sum throughput vs the codeword length.}
\vspace{-5mm}
\end{figure}

\footnotesize
\vspace{-3mm}
\bibliographystyle{IEEEtran} 
\vspace{-1mm}
\bibliography{ref_letter}

\begin{thebibliography}{1}
\providecommand{\url}[1]{#1}
\csname url@samestyle\endcsname
\providecommand{\newblock}{\relax}
\providecommand{\bibinfo}[2]{#2}
\providecommand{\BIBentrySTDinterwordspacing}{\spaceskip=0pt\relax}
\providecommand{\BIBentryALTinterwordstretchfactor}{4}
\providecommand{\BIBentryALTinterwordspacing}{\spaceskip=\fontdimen2\font plus
\BIBentryALTinterwordstretchfactor\fontdimen3\font minus
  \fontdimen4\font\relax}
\providecommand{\BIBforeignlanguage}[2]{{%
\expandafter\ifx\csname l@#1\endcsname\relax
\typeout{** WARNING: IEEEtran.bst: No hyphenation pattern has been}%
\typeout{** loaded for the language `#1'. Using the pattern for}%
\typeout{** the default language instead.}%
\else
\language=\csname l@#1\endcsname
\fi
#2}}
\providecommand{\BIBdecl}{\relax}
\BIBdecl

\bibitem{tutor2}
F.~Boccardi, R.~W. Heath, A.~Lozano, T.~L. Marzetta, and P.~Popovski, ``Five
  disruptive technology directions for {5G},'' \emph{IEEE Commun. Mag.},
  vol.~52, no.~2, pp. 74--80, Feb. 2014.

\bibitem{poly_main}
Y.~Polyanskiy, H.~V. Poor, and S.~Verdu, ``Channel coding rate in the finite
  blocklength regime,'' \emph{IEEE Trans. Inf. Theory}, vol.~56, no.~5, pp.
  2307 -- 2359, May 2010.

\bibitem{cog_mak}
B.~Makki, T.~Svensson, and M.~Zorzi, ``Finite block-length analysis of spectrum
  sharing networks using rate adaptation,'' \emph{IEEE Trans. Commun.},
  vol.~63, no.~8, pp. 2823--2835, Aug. 2015.

\bibitem{li2016throughput}
Y.~Li, M.~C. Gursoy, and S.~Velipasalar, ``Throughput of two-hop wireless
  channels with queueing constraints and finite blocklength codes,'' in
  \emph{Proc. {IEEE} ISIT'2016}, Barcelona, Spain, July 2016, pp. 2599--2603.

\bibitem{harq}
B.~Makki, T.~Svensson, and M.~Zorzi, ``Finite block-length analysis of the
  incremental redundancy {HARQ},'' \emph{IEEE Wireless Commun. Lett.}, vol.~3,
  no.~5, pp. 529--532, Oct. 2014.

\bibitem{bert}
D.~P. Bertsekas, \emph{Nonlinear programming}.\hskip 1em plus 0.5em minus
  0.4em\relax Athena scientific Belmont, 1999.

\bibitem{tan2015third}
V.~Y.~F. Tan and M.~Tomamichel, ``The third-order term in the normal
  approximation for the awgn channel,'' \emph{IEEE Trans. Inf. Theory},
  vol.~61, no.~5, pp. 2430--2438, 2015.

\bibitem{yu2016tradeoff}
W.~Yu, L.~Musavian, and Q.~Ni, ``Tradeoff analysis and joint optimization of
  link-layer energy efficiency and effective capacity toward green
  communications,'' \emph{IEEE Trans. Wireless Commun.}, vol.~15, no.~5, pp.
  3339--3353, May 2016.

\bibitem{palomar2006tutorial}
D.~P. Palomar and M.~Chiang, ``A tutorial on decomposition methods for network
  utility maximization,'' \emph{IEEE J. Sel. Areas Commun.}, vol.~24, no.~8,
  pp. 1439--1451, Aug. 2006.

\end{thebibliography}

\end{document}